\documentclass[11pt]{article}
\usepackage[margin=1in]{geometry}
\usepackage[final]{microtype}
\usepackage{epsfig}
\usepackage{graphics}
\usepackage{latexsym}
\usepackage{amsmath}
\usepackage{amsfonts}
\usepackage{amssymb}
\usepackage{mathrsfs}
\usepackage[dvipsnames]{xcolor}
\usepackage{amsthm}
\usepackage{xspace}
\usepackage{epstopdf}
\usepackage{float}
\usepackage{hyperref}
\usepackage{pgfplots}
\pgfplotsset{compat=1.18}
\usepackage{longtable}

\numberwithin{equation}{section}
\usepackage[ruled,vlined]{algorithm2e}
\SetArgSty{textrm}

\usepackage[english]{babel}
\usepackage[nottoc]{tocbibind}

\usepackage{caption}
\usepackage{subcaption}

\usepackage{pifont}
\usepackage{booktabs}
\usepackage{multirow}

\theoremstyle{plain}
\newtheorem{theorem}{Theorem}

\newtheorem{lemma}{Lemma}

\theoremstyle{definition}

\theoremstyle{remark}

\makeatletter
\@addtoreset{equation}{section}
\def\section{\@startsection {section}{1}{\z@}{-3.5ex plus -1ex minus
 -.2ex}{2.3ex plus .2ex}{\large\bf}}
\makeatother

\def\bfm#1{\mbox{\boldmath$#1$}}

\def\0{\bfm 0}

\DeclareMathAlphabet{\mathpzc}{OT1}{pzc}{m}{it}
\newcommand{\argmin}{\operatorname*{arg\,min}}

\newcommand{\OPT}{\operatorname{OPT}}

\newcounter{my}

\newcounter{my2}

\newcounter{my3}

\newcounter{my4}

\newcounter{my5}

\newcounter{my6}

\allowdisplaybreaks

\begin{document}

\title{Optimality and Impossibility in Multidimensional Strategyproof Facility Location}
\author{Jianan Lin\\[0.75em]
Rensselaer Polytechnic Institute
}
\date{}
\maketitle

\begin{abstract}
We study facility location in multidimensional normed spaces with the objective of minimizing social cost, the sum of agents' distances to their nearest facility.
A deterministic mechanism selects $k$ facility locations based on the reports of $n$ agents and is strategyproof (SP) if no agent can reduce her distance by misreporting her location.
We establish optimal approximation ratios and impossibility results.
For a single facility in the Euclidean plane, we prove that the optimal approximation ratios among deterministic SP mechanisms are exactly those obtained by Goel and Hann-Caruthers [SCW 2023] for the coordinate-wise median (CM): $\frac{\sqrt{2(n^2+1)}}{n+1}$ for odd $n\ge3$ and $\sqrt{2}$ for even $n\ge4$.
We remove the anonymity assumption from their optimality result.
However, CM is not always optimal in non-Euclidean spaces: we give an SP mechanism that minimizes social cost exactly in $\ell_\infty^d$, whereas CM has a tight uniform approximation ratio of $2$ already in the plane.
We then consider two facilities in $\ell_p$ spaces with $1<p<\infty$ and prove two impossibility results for every $n>2$.
In the plane, no deterministic SP mechanism has a finite approximation ratio.
In higher dimensions, the same impossibility holds for anonymous deterministic SP mechanisms.
Finally, a reduction using distant agents and a limiting argument extends both impossibility results to every $k>2$ with $n>k$.
\end{abstract}

\section{Introduction}\label{sec:introduction}

Facility location is a basic problem in mechanism design without money and social choice~\cite{moulin1980strategy,schummer2007mechanism,procaccia2013approximate}.
A planner chooses locations for public facilities using information supplied by the people who will use them.
Each agent prefers a facility close to her own location and may benefit from reporting a different location to influence the outcome.
A strategyproof mechanism makes truthful reporting optimal regardless of the other reports.
The planner also seeks to serve the population efficiently, commonly measured by social cost, the sum of distances from agents to their nearest facilities.
Understanding how much efficiency must be sacrificed to ensure strategyproofness is a central question in this setting.

For a single facility, classical characterization results explain the role of median rules.
Moulin~\cite{moulin1980strategy} developed the generalized-median characterization for single-peaked preferences on the line, where the ordinary median is both strategyproof and socially optimal.
In the Euclidean plane, Peters, van der Stel, and Storcken~\cite{peters1993range} connected strategyproofness with continuity and convexity of the range and obtained characterizations under anonymity.
Procaccia and Tennenholtz~\cite{procaccia2013approximate} made facility location a central example of approximate mechanism design without money, focusing on the efficiency guarantees achievable subject to incentive constraints.
In multidimensional spaces, taking a median in each coordinate remains strategyproof, but can increase social cost above the optimum.
Meir~\cite{DBLP:conf/sagt/Meir19} established a $\sqrt{d}$ approximation upper bound for the coordinate-wise median in $d$-dimensional Euclidean space for any number of agents, and determined its exact ratio of $\frac{\sqrt{5}}{2}$ for three agents in the plane.
Goel and Hann-Caruthers~\cite{GoelH23} determined its exact approximation ratio for odd populations and proved its optimality among anonymous deterministic strategyproof mechanisms in the plane.
More recently, Gravin and Jia~\cite{gravin2025approximation} established dimension-independent approximation bounds for the coordinate-wise median in $\ell_p$ spaces, with matching bounds as the dimension grows.
Hastings~\cite{hastings2026strategic} analyzed more general norm-based social costs, obtaining as a special case the tight uniform approximation ratio $2^{1-\frac{1}{p}}$ for CM in $\ell_p^2$.

With multiple facilities, each agent's cost depends on the nearest selected location, and even the problem on a line has substantially different approximation guarantees.
For two facilities, Procaccia and Tennenholtz~\cite{procaccia2013approximate} studied the mechanism that selects the leftmost and rightmost reports.
Its social-cost approximation ratio is $n-2$.
Lu et al.~\cite{lu10mechanism} proved a lower bound of $\frac{n-1}{2}$ for deterministic strategyproof mechanisms on the line.
Fotakis and Tzamos~\cite{fotakis2014power} subsequently obtained the tight ratio of $n-2$ for $n\ge5$ through a characterization of mechanisms with finite approximation guarantees.
They also proved that anonymous deterministic strategyproof mechanisms for three or more facilities cannot achieve a finite approximation ratio on the line, already for instances with one more agent than facilities.
Thus, even on the line, adding facilities can turn an exactly solvable problem into one with a population-dependent approximation ratio or no finite guarantee.

In multidimensional normed spaces, several lower-bound questions remain open.
Jia~\cite{jia2026aggregation} recently proved that the coordinate-wise median is optimal among continuous, anonymous, deterministic strategyproof mechanisms in every finite-dimensional Euclidean space, for odd $n\ge3$ and even $n\ge4$.
Together with the planar results, this raises the question of whether anonymity is necessary for median optimality and whether that optimality persists under other spatial norms.
For multiple facilities, Fotakis and Tzamos~\cite{fotakis2014power} proved that even two facilities admit no finite deterministic strategyproof approximation on a star with three agents.
Extending these lower bounds to normed spaces requires accounting for facilities placed outside the line or star supporting the reports.
Can deterministic strategyproof mechanisms achieve a finite social-cost approximation in multidimensional normed spaces, and how does the answer depend on anonymity and the norm?

\subsection{Our Contributions}

We study deterministic SP mechanisms for $n$ agents and $k$ facilities in multidimensional normed spaces, focusing on $\ell_p^d$ and the social-cost objective.
Our main contributions are as follows.

\begin{itemize}
    \item \textbf{Single-facility optimality.}
    We prove that CM is optimal among all deterministic SP mechanisms in the Euclidean plane for every $n\ge3$.
    Our lower bounds match its known approximation ratios~\cite{GoelH23}: $\frac{\sqrt{2(n^2+1)}}{n+1}$ for odd $n\ge3$ and $\sqrt{2}$ for even $n\ge4$.
    We remove the anonymity assumption from their optimality result.
    For $n=2$, a dictatorship attains ratio $1$.
    The proof groups agents at three locations and uses SP to constrain the mechanism's coordinate choices.
    In non-Euclidean spaces, CM is not always optimal: we show that selecting the lexicographically smallest social-cost minimizer is anonymous and SP in $\ell_\infty^d$ for every finite $d\ge2$.
    This achieves ratio $1$, improving on CM's tight uniform ratio of $2$ already in the plane~\cite{hastings2026strategic}.

    \item \textbf{Two-facility impossibility.}
    For every $1<p<\infty$ and every fixed $n\ge3$, we prove that no deterministic SP two-facility mechanism in $\ell_p^2$ has a finite approximation ratio.
    In every finite dimension $d\ge2$, we obtain the same impossibility for anonymous mechanisms.
    The higher-dimensional proof needs only invariance under exchanging one fixed pair of agents.
    We analyze the facility locations attainable by changing a single report and use the resulting geometric restrictions to derive a contradiction.

    \item \textbf{Extension to multiple facilities.}
    We extend both impossibility results to every $3\le k<n$, with the same respective assumptions on anonymity.
    The proof uses a reduction that converts a deterministic SP $\rho$-approximation for $n$ agents and $k$ facilities into one for $n-k+2$ agents and two facilities.
    By placing $k-2$ agents increasingly far away and taking a common limit of the remaining outputs, the reduction preserves SP, the approximation guarantee, and any symmetry among the retained agents.
    The reduction applies to every finite-dimensional normed space.
\end{itemize}

\subsection{Related Work}

\paragraph{Single-Facility Location in Different Metric Spaces.}
Beyond the line and Euclidean spaces discussed above, \cite{schummer2002strategy} studied strategyproof location on networks, while \cite{dokow2012mechanism} considered discrete lines and cycles.
For the maximum-cost objective, \cite{alon2010strategyproof} established approximation bounds on networks.
On a circle, \cite{DBLP:conf/sagt/Meir19} gave a randomized $\frac{7}{6}$-approximation for three agents, and \cite{rogowski2025improved} improved the guarantee for larger odd populations by mixing random dictatorship with the proportional circle distance mechanism.
More recently, \cite{chan2026circle} established a randomized SP $\frac{3}{2}$-approximation for every $n\ge3$ on the circle.
In normed spaces, \cite{lin2020nearly} studied the characterization of deterministic SP mechanisms for two agents, and \cite{tang2020characterization} characterized group-strategyproof mechanisms in strictly convex spaces.
For Euclidean social cost, \cite{barak2026facility,chan2026strategyproof} independently obtained a randomized $\frac{4}{\pi}$-approximation in the plane by randomly rotating the coordinate system before applying the coordinate-wise median.
Related work also considers alternative ways of aggregating agents' distances, including least-squares cost on lines and trees~\cite{feldman2013strategyproof} and $L_p$-norm social cost on the line~\cite{feigenbaum2017approximately}, in the Euclidean plane~\cite{chan2026strategyproof}, and in more general normed spaces~\cite{hastings2026strategic}.

\paragraph{Multiple-Facility Location.}
In addition to the deterministic results reviewed above, \cite{escoffier2011many} studied approximation for multiple facilities under total and maximum cost, and \cite{lu2024design} investigated the structure and design of dictatorial two-facility mechanisms on the line.
Randomization gives substantially better social-cost guarantees for two facilities: \cite{lu2009tighter} established early upper and lower bounds, and \cite{lu10mechanism} gave the Proportional mechanism with approximation ratio $4$ on general metric spaces.
More recently, \cite{ma2026breaking} improved the upper bound to $\frac{11}{3}$ on Ptolemaic spaces, which include Euclidean spaces, and the lower bound to $\frac{1+\sqrt{2}}{2}$ even on the line.
The ratio was further improved to $\frac{74+4\sqrt{3}}{23}\approx 3.5186$ independently in \cite{jia2026product} for the line and \cite{chan2026randomized} for Ptolemaic spaces.
For three facilities on the line, \cite{aziz2026anchoring} introduced the Random-Anchor Volume mechanism and proved that it is strategyproof in expectation with approximation ratio $8$.
This ratio was subsequently improved to $6$ in both \cite{jia2026product} and \cite{chan2026randomized}.
For $k=n-1$ facilities, \cite{chan2026randomized} gave a randomized $\Theta(\sqrt{n})$-approximation on arbitrary metric spaces.
Other variants include concave connection costs~\cite{fotakis2013strategyproof} and winner-imposing mechanisms, where agents selected to host facilities must use those facilities~\cite{fotakis2013winner}.

\section{Preliminaries}\label{sec:preliminaries}

Let $N=[n]=\{1,\ldots,n\}$ be a set of agents in a normed space $X=(\mathbb{R}^d,\|\cdot\|)$, where $d\ge2$.
Each agent $i$ has a private location $x_i\in X$, and we write $\mathbf{x}=(x_1,\ldots,x_n)$ for the location profile.
We mainly consider $\ell_p^d$ spaces, where
\[
    \|v\|_p=\left(\sum_{r=1}^d|v_r|^p\right)^{1/p}\quad(1\le p<\infty),
    \qquad
    \|v\|_\infty=\max_{r\in[d]}|v_r|.
\]
The Euclidean case corresponds to $p=2$.

A deterministic $k$-facility mechanism maps each reported profile to a nonempty set $f(\mathbf{x})\subseteq X$ of at most $k$ facility locations.
Facilities are identical, and each agent uses a nearest facility.
Representing the output as a set allows facilities to coincide, since repeated facilities do not affect any agent's cost.
There are no payments.
For a facility set $Y$, agent $i$'s cost and the social cost are
\[
    c_i(\mathbf{x},Y)=\min_{y\in Y}\|x_i-y\|,
    \qquad
    \operatorname{SC}(\mathbf{x},Y)=\sum_{i=1}^n c_i(\mathbf{x},Y).
\]
We focus on $1\le k<n$, since placing a facility at every distinct report gives zero cost when $k\ge n$.

A mechanism is \emph{strategyproof} (SP) if no agent can reduce her cost by misreporting her location.
Formally, for every profile $\mathbf{x}$, every agent $i\in N$, and every alternative report $x_i'\in X$,
\[
    c_i\bigl(\mathbf{x},f(\mathbf{x})\bigr)
    \le
    c_i\bigl(\mathbf{x},f(x_i',\mathbf{x}_{-i})\bigr),
\]
where $(x_i',\mathbf{x}_{-i})$ denotes the profile obtained by replacing only agent $i$'s report.
Both costs are evaluated at her true location $x_i$.
A mechanism is \emph{anonymous} if permuting the agents' reports leaves its facility set unchanged: $f(x_{\pi(1)},\ldots,x_{\pi(n)})=f(\mathbf{x})$ for every permutation $\pi$ of $N$.
We impose anonymity only when explicitly stated.

The optimal social cost of a profile is
\[
    \OPT_k(\mathbf{x})
    =\min_{\substack{Y\subseteq X\\1\le |Y|\le k}}
      \operatorname{SC}(\mathbf{x},Y).
\]
A mechanism $f$ is a $\rho$-approximation if
\[
    \operatorname{SC}\bigl(\mathbf{x},f(\mathbf{x})\bigr)
    \le \rho\,\OPT_k(\mathbf{x})
    \qquad\text{for every }\mathbf{x}\in X^n.
\]
Its approximation ratio $\operatorname{AR}(f)$ is the infimum of such $\rho\ge1$, and is infinite if no finite $\rho$ satisfies this inequality.
In particular, a finite approximation guarantee requires zero mechanism cost whenever $\OPT_k(\mathbf{x})=0$.
An unbounded approximation ratio means that no finite guarantee holds over all location profiles, even with the parameters $n$, $k$, and the space $X$ fixed.

\section{Single-Facility Location}\label{sec:single-facility}

For a single facility, we write the mechanism's output as a point in the underlying space.
The coordinate-wise median (CM) chooses a median of the reports in each coordinate, using the lower median when $n$ is even.
We prove its optimality in the Euclidean plane without anonymity, then give an exact SP mechanism for $\ell_\infty$ spaces.

\subsection{Lower Bounds in the Euclidean Plane}\label{sec:single-euclidean}

Throughout this subsection, CM uses a fixed orthonormal coordinate system.
Goel and Hann-Caruthers~\cite{GoelH23} established its optimality among anonymous SP mechanisms and determined its exact approximation ratio for odd $n$.
The following theorem removes the anonymity requirement.

\begin{theorem}\label{thm:planar-optimality}
For every $n\ge3$, the optimal approximation ratio of a deterministic SP single-facility mechanism in $\ell_2^2$ is
\[
    \alpha_n=\sqrt{1+\left(\frac{\lfloor n/2\rfloor}{\lceil n/2\rceil}\right)^2}
    =\begin{cases}
        \frac{\sqrt{2(n^2+1)}}{n+1},&\text{if $n$ is odd},\\[3pt]
        \sqrt{2},&\text{if $n$ is even}.
    \end{cases}
\]
The coordinate-wise median attains this ratio.
For $n=2$, the optimal ratio is $1$.
\end{theorem}

Suppose a mechanism has approximation ratio below $\alpha_n$.
We group its agents at three locations and show that the induced rule agrees with CM on a family of profiles.
Choosing the group sizes and locations then gives the required lower bound.
Throughout the Euclidean argument, $\|\cdot\|$ denotes Euclidean distance and $\langle\cdot,\cdot\rangle$ its inner product.

A planar Euclidean SP mechanism with convex range is continuous~\cite{peters1993range}.
A finite approximation ratio implies $f(a,\ldots,a)=a$, so the range is the whole plane and this continuity theorem applies.
For a continuous SP mechanism, the response to one report is the nearest-point projection onto its closed convex range~\cite[Lemma 2.1]{jia2026aggregation}.
Thus, writing
\[
    M_i(\mathbf{x}_{-i})=\{f(a,\mathbf{x}_{-i}):a\in\mathbb{R}^2\},
\]
we have $f(a,\mathbf{x}_{-i})=P_{M_i(\mathbf{x}_{-i})}(a)$, where $P_C$ denotes Euclidean projection onto $C$.
The set $M_i(\mathbf{x}_{-i})$ is the agent's unilateral menu, also called an image set in facility location~\cite[Section 3.1]{lu10mechanism}.
In particular, $f$ is $1$-Lipschitz in each report.
Changing reports one at a time therefore gives
\[
    \|f(\mathbf{x})-f(\mathbf{x}')\|\le\sum_{i=1}^n\|x_i-x_i'\|.
\]
The projection inequality also gives radial invariance~\cite[Corollary 2.3]{jia2026aggregation}: if $f(\mathbf{x})=y$, then
\begin{equation}\label{eq:single-radial}
    f\bigl(y+t_1(x_1-y),\ldots,y+t_n(x_n-y)\bigr)=y
    \qquad(t_1,\ldots,t_n\ge0).
\end{equation}
Indeed, $\langle x_i-y,z-y\rangle\le0$ for every $z\in M_i(\mathbf{x}_{-i})$, and multiplying $x_i-y$ by a nonnegative scalar preserves this inequality.
The reports can be changed one at a time while the output remains $y$.
The next lemma uses these properties to normalize the mechanism while preserving its approximation guarantee.

\begin{lemma}[Normalization]\label{lem:single-normalization}
Every planar Euclidean SP mechanism with finite approximation ratio has a continuous SP replacement $\bar f$ with no larger approximation ratio such that $\bar f(\mathbf{x}+v)=\bar f(\mathbf{x})+v, \quad \bar f(t\mathbf{x})=t\bar f(\mathbf{x}) \quad(v\in\mathbb{R}^2,\ t\ge0)$. Here $\mathbf{x}+v$ adds $v$ to every report, and $t\mathbf{x}$ multiplies every report by $t$.
\end{lemma}
\begin{proof}
The preceding observations give continuity and unanimity.
We use the construction in Jia~\cite[Theorem 7.1 and Appendix A]{jia2026aggregation}.
The limit construction uses unanimity, continuity, and the projection and radial-invariance properties above.
In Jia's proof, anonymity is verified separately as a property inherited by the limit, so the same construction applies to the mechanisms considered here.

The construction gives a direction $u\ne0$ such that, for every profile $\mathbf{x}$, the limit
\[
    L(\mathbf{x})=\lim_{t\to\infty}\bigl(tu+v-f(\mathbf{x}+tu+v)\bigr)
\]
exists and is independent of the fixed offset $v\in\mathbb{R}^2$.
Appendix A first establishes this conclusion on a countable dense set of profiles, using a bounded monotone displacement map and a common exposed direction.
The uniform Lipschitz bound then extends convergence to every profile, locally uniformly.
Consequently, $f_t(\mathbf{x})=f(\mathbf{x}+tu)-tu$ converges to a continuous mechanism $\bar f(\mathbf{x})=-L(\mathbf{x})$.
Offset independence gives
\[
    \bar f(\mathbf{x}+v) =\lim_{t\to\infty}\bigl(f(\mathbf{x}+v+tu)-tu\bigr) =v-L(\mathbf{x})=\bar f(\mathbf{x})+v.
\]

Each $f_t$ is a translated copy of $f$, so it is SP and has the same approximation guarantee.
Passing to the limit in the incentive inequality for each fixed true report and deviation preserves SP.
Likewise, for $\rho=\operatorname{AR}(f)$,
\[
    \operatorname{SC}(\mathbf{x},f_t(\mathbf{x}))\le\rho\OPT_1(\mathbf{x}) \quad\Longrightarrow\quad \operatorname{SC}(\mathbf{x},\bar f(\mathbf{x}))\le\rho\OPT_1(\mathbf{x}).
\]
Finally, put $y=\bar f(\mathbf{x})$.
Radial invariance and translation equivariance give
\[
    y=\bar f(s\mathbf{x}+(1-s)y)=\bar f(s\mathbf{x})+(1-s)y \qquad(s\ge0),
\]
so $\bar f(s\mathbf{x})=sy$, as in~\cite[Corollary 2.4]{jia2026aggregation}.
\end{proof}

To group agents with identical reports, we use \emph{partial group strategyproofness}: agents at a common true location cannot reduce their common cost through a joint misreport.
Lu et al.~\cite[Definition 2.3 and Lemma 2.1]{lu10mechanism} showed that individual SP implies this property for general $k$-facility games.
For a set $S\subseteq N$, the notation $a^S$ means that each member of $S$ reports $a$, with identities retained.
The induced rule remains SP with respect to the common report of a group.
To see this, change its members' reports from $a$ to $b$ one at a time.
At each step, the changing agent's SP constraint, evaluated at the true location $a$, says that the distance from $a$ to the output cannot decrease.
Chaining these inequalities proves the claim.
Continuity and the two invariances in Lemma~\ref{lem:single-normalization} are also preserved under grouping.

The next lemma restricts the mechanism on profiles supported at two locations.
For each fixed partition of the agents, the same group must determine the output at every pair of locations.

\begin{lemma}\label{lem:single-binary}
Let $f$ be a continuous, translation-equivariant, positively homogeneous SP mechanism in the Euclidean plane with $\operatorname{AR}(f)<\sqrt{2}$.
For every nonempty proper subset $S\subsetneq N$, either $f(a^S,b^{N\setminus S})=a$ for all $a,b\in\mathbb{R}^2$, or $f(a^S,b^{N\setminus S})=b$ for all such $a,b$.
\end{lemma}
\begin{proof}
Set $h(a,b)=f(a^S,b^{N\setminus S})$ and $s=|S|$.
The projection property applies to this two-group SP rule.
Its range $K=\{h(a,0):a\in\mathbb{R}^2\}$ is closed and convex by the projection property.
Positive homogeneity gives $h(ta,0)=t h(a,0)$ for $t\ge0$, so $K$ is a cone.
Translating both reports by $-b$ now gives
\[
    h(a,b)=b+P_K(a-b).
\]
Suppose that $K$ is neither $\{0\}$ nor $\mathbb{R}^2$.
There exists $z$ for which $p=P_Kz$ and $q=z-p$ are both nonzero.
Otherwise every point would belong to either $K$ or its polar cone $K^\circ=\{v:\langle v,w\rangle\le0\text{ for all }w\in K\}$.
These are closed cones with intersection $\{0\}$, and both contain nonzero points, contradicting the connectedness of $\mathbb{R}^2\setminus\{0\}$.
The projection inequality, tested at $0$ and $2p$, gives $\langle p,q\rangle=0$.

Since $h(z,0)=p$, translation gives $h(q,-p)=0$.
Set $e_1=\frac{q}{\|q\|}$ and $e_2=-\frac{p}{\|p\|}$.
These are orthogonal unit vectors.
Scaling each group's report about the output zero gives $h(re_1,te_2)=0$ for all $r,t>0$ by radial invariance.
The mechanism cost on this profile is $sr+(n-s)t$.
For two reported locations, the optimum places the facility at a location with at least as many agents, and its cost is $\min\{s,n-s\}\sqrt{r^2+t^2}$.
It follows that
\[
    \operatorname{AR}(f) \ge\sup_{r,t>0}\frac{sr+(n-s)t}{\min\{s,n-s\}\sqrt{r^2+t^2}} =\frac{\sqrt{s^2+(n-s)^2}}{\min\{s,n-s\}} \ge\sqrt{2},
\]
a contradiction.
Thus $K=\{0\}$ or $K=\mathbb{R}^2$, giving the two alternatives.
\end{proof}

Call $S$ \emph{winning} if the first alternative of Lemma~\ref{lem:single-binary} holds.
Exactly one of $S$ and $N\setminus S$ is winning.
Take a winning set $T$ of minimum cardinality $t$.
For $n\ge3$, a singleton cannot be winning: placing it at $a$ and all other agents at $b\ne a$ would give ratio $n-1\ge2$.
A partition into two sets of sizes $\lfloor n/2\rfloor$ and $\lceil n/2\rceil$ shows that $t\le\lceil n/2\rceil$.
Split $T$ into two nonempty sets $A,B$, and put $C=N\setminus T$.
Minimality of $T$ makes $A$ and $B$ losing, and $C$ is losing because its complement $T$ is winning.
The complement of each losing group is winning, so every union of two groups is winning.
The three-group rule $H(a,b,c)=f(a^A,b^B,c^C)$ therefore satisfies
\begin{equation}\label{eq:single-majority}
    H(a,a,b)=H(a,b,a)=H(b,a,a)=a.
\end{equation}
We will show that these majority identities force $H$ to agree with CM in a fixed coordinate system on a family of profiles.
We first show that the derivative matrices with respect to the agents' reports share an orthonormal eigenbasis.

\begin{lemma}\label{lem:single-commuting}
Let $F:(\mathbb{R}^2)^m\to\mathbb{R}^2$ be continuous and SP.
At almost every profile, its derivative blocks $A_i=D_iF$ are symmetric positive semidefinite contractions and commute pairwise: $A_iA_j=A_jA_i$.
\end{lemma}
\begin{proof}
Each unilateral response is a convex projection.
Such a projection is the gradient of the convex function $a\mapsto\frac{1}{2}\|a\|^2-\frac{1}{2}\operatorname{dist}(a,M)^2$ and is firmly nonexpansive.
Its derivative is therefore symmetric and satisfies $0\preceq A_i\preceq I$ wherever it exists~\cite[Lemma A.1]{jia2026aggregation}.
Here $A\preceq B$ means that $B-A$ is positive semidefinite.
The mechanism is Lipschitz, so Rademacher's theorem gives differentiability almost everywhere.

For pairwise commutation, we adapt the differentiation argument in~\cite[Appendix A.2]{jia2026aggregation} by moving a single report radially.
Fix $k$, and let $R_k$ replace report $x_k$ by $2x_k-F(\mathbf{x})$, leaving the other reports fixed.
Radial invariance gives $F(R_k\mathbf{x})=F(\mathbf{x})$.
Moreover, $R_k$ is bi-Lipschitz: its inverse replaces $x_k$ by $\frac{x_k+F(\mathbf{x})}{2}$.
A bi-Lipschitz map and its inverse preserve sets of measure zero.
Removing the nondifferentiability set of $F$ and its preimages under the finitely many maps $R_k$, we can differentiate at both $\mathbf{x}$ and every $R_k\mathbf{x}$.
Writing $B_j=D_jF(R_k\mathbf{x})$, the chain rule gives
\[
    B_k(2I-A_k)=A_k,\qquad B_j=(I+B_k)A_j\quad(j\ne k).
\]
Hence $I+B_k=(I-\frac{1}{2} A_k)^{-1}$.
The inverse exists because the eigenvalues of $A_k$ lie in $[0,1]$.
Since $B_j$ is symmetric, transposing $B_j=(I-\frac{1}{2}A_k)^{-1}A_j$ gives
\[
    \left(I-\frac{1}{2} A_k\right)^{-1}A_j=A_j\left(I-\frac{1}{2} A_k\right)^{-1}.
\]
Taking the inverse and rearranging yields $A_kA_j=A_jA_k$.
\end{proof}

\begin{lemma}\label{lem:single-three-groups}
Let $H:(\mathbb{R}^2)^3\to\mathbb{R}^2$ be continuous, translation-equivariant, and SP, and satisfy~\eqref{eq:single-majority}.
There exist orthogonal unit vectors $e_1,e_2$ such that $H(-re_1,re_1,e_2)=0$ for every $r>0$.
\end{lemma}
\begin{proof}
At a profile $\mathbf{x}$, write $y=H(\mathbf{x})$ and $r_i=x_i-y$.
For example, agent $1$ can induce $x_2$ by reporting $x_2$ and can induce $x_3$ by reporting $x_3$, by~\eqref{eq:single-majority}.
Thus each agent's menu contains both other reports.
The projection inequality implies
\begin{equation}\label{eq:single-obtuse}
    \langle r_i,r_j\rangle\le0\qquad(i\ne j).
\end{equation}
Consider the open set of profiles forming acute triangles.
Here every $r_i$ is nonzero: if $y$ were a vertex, the other two residual vectors would form an acute angle, contradicting~\eqref{eq:single-obtuse}.

At almost every such profile, let $A_i=D_iH$.
Translation equivariance and radial invariance imply, respectively,
\[
    A_1+A_2+A_3=I,\qquad A_ir_i=0\quad(i=1,2,3).
\]
Each $A_i$ has rank at most one.
By Lemma~\ref{lem:single-commuting}, these matrices are simultaneously diagonalizable in an orthonormal basis.
Their sum is $I$, so two nonzero blocks must act on different axes.
Their kernel directions, and hence the corresponding residual vectors, are perpendicular.
The differentiability points have full measure and are dense in the acute-profile domain.
At each such point one of the three residual pairs is perpendicular.
Approaching any acute profile along a sequence and retaining a subsequence with the same pair, continuity preserves its zero inner product.

A perpendicular pair determines an orthogonal frame, with the order and signs of its axes ignored.
This frame is unique at each acute profile.
Indeed, two perpendicular pairs must share an index, and their other vectors are parallel, giving the same frame.
We next show that the axes of this frame are locally constant.

Fix an acute profile.
Choose an index $i$ contained in every perpendicular pair there, and denote the other indices by $j,k$.
Such an index exists because three nonzero vectors in the plane cannot be pairwise perpendicular.
We have $\langle r_j,r_k\rangle<0$, and this remains true in a neighborhood.
Every perpendicular pair in that neighborhood therefore contains $i$.
Thus the direction of $r_i$ determines the frame throughout this neighborhood.
After shrinking the neighborhood if necessary, choose a continuous direction angle $\theta$ for $r_i$.
Since $H$ is Lipschitz and $r_i\ne0$, this angle is locally Lipschitz.

Let $P_{r_i}$ denote projection onto the line spanned by $r_i$, and $P_{r_i^\perp}=I-P_{r_i}$.
At almost every point in this neighborhood, the derivative blocks have the form
\[
    A_i=P_{r_i^\perp},\qquad A_j=\lambda P_{r_i},\qquad A_k=(1-\lambda)P_{r_i},\qquad 0\le\lambda\le1.
\]
To verify this, use the common eigenbasis above.
The nonzero blocks have kernels parallel to their residual vectors.
Since $r_j,r_k$ are not perpendicular, they cannot supply blocks acting on both distinct axes.
In particular, $A_i$ cannot be zero, since then $A_j+A_k=I$ would require $r_j\perp r_k$.
Every other nonzero block acts along $r_i$: a block acting perpendicular to $r_i$ would have its residual parallel to $r_i$, contradicting the choice that every perpendicular pair contains $i$.
Since the blocks sum to $I$, $A_i$ is the full projection perpendicular to $r_i$, while $A_j$ and $A_k$ divide the projection onto $r_i$.
It follows that
\[
    D_i r_i=P_{r_i},\qquad D_j r_i=-\lambda P_{r_i},\qquad D_k r_i=-(1-\lambda)P_{r_i}.
\]
All first-order changes of $r_i$ are parallel to $r_i$, so its direction angle satisfies $D\theta=0$ almost everywhere.
A locally Lipschitz function with zero derivative almost everywhere is constant on a sufficiently small open ball.
The frame is therefore locally constant, and hence constant on each connected component of the acute-profile domain.

There are two such components, distinguished by the orientation of the ordered triangle.
To see connectedness for positive orientation, parameterize the first two vertices as $(0,0)$ and $(1,0)$ after translation, rotation, and positive scaling.
Writing the third vertex as $(u,v)$, the three acute-angle conditions become
\[
    0<u<1,\qquad v>\sqrt{u(1-u)}.
\]
This parameter domain is connected, as are the translation, rotation, and scale parameters.
Fix the frame on the positively oriented component.
At any profile in this component, express a perpendicular residual pair as $(a,0)$ and $(0,b)$ in that frame, and the remaining vector as $(c,d)$.
Equation~\eqref{eq:single-obtuse} gives $ac\le0$ and $bd\le0$.
Thus zero is the median of both residual coordinates, and $H$ equals CM on this component.

Choose unit axes $e_1,e_2$ of this frame with positive orientation.
For $0<r_0<1$, the profile $(-r_0e_1,r_0e_1,e_2)$ forms a positively oriented acute triangle and has coordinate-wise median zero.
Hence $H(-r_0e_1,r_0e_1,e_2)=0$.
Scaling the first two reports about zero using~\eqref{eq:single-radial} gives the claimed identity for every $r>0$.
\end{proof}

\begin{proof}[Proof of Theorem~\ref{thm:planar-optimality}]
Suppose $n\ge3$ and an SP mechanism has ratio strictly below $\alpha_n$.
By Lemma~\ref{lem:single-normalization}, we may assume continuity, translation equivariance, and positive homogeneity without increasing its ratio.
Since $\alpha_n\le\sqrt{2}$, Lemma~\ref{lem:single-binary} applies.
Choose $T,A,B,C$ as above, with $t=|T|\le\lceil n/2\rceil$ and $|C|=n-t$.
Lemma~\ref{lem:single-three-groups} supplies a profile with reports $-re_1$ on $A$, $re_1$ on $B$, and $e_2$ on $C$, on which the mechanism selects zero.
Its social cost is $tr+n-t$.
Placing the facility at $e_2$ serves $C$ at zero cost and each member of $A\cup B$ at distance $\sqrt{r^2+1}$, so $\OPT_1(\mathbf{x})\le t\sqrt{r^2+1}$.
Taking $r=\frac{t}{n-t}$ gives
\[
    \operatorname{AR}(f)
    \ge\frac{tr+n-t}{t\sqrt{r^2+1}}
    =\sqrt{1+\left(\frac{n-t}{t}\right)^2}
    \ge\sqrt{1+\left(\frac{\lfloor n/2\rfloor}{\lceil n/2\rceil}\right)^2}
    =\alpha_n,
\]
a contradiction.

CM is SP because each coordinate of a truthful report is projected onto the interval of attainable coordinate medians.
For odd $n$, its approximation ratio is $\alpha_n$ by Goel and Hann-Caruthers~\cite{GoelH23}.
For even $n$, the ratio $\sqrt{2}$ is also noted in~\cite[Footnote 2]{GoelH23}.
The upper bound follows directly from norm comparison. Let $y^*$ be a Euclidean optimum and write $\operatorname{SC}_p$ for the sum of $\ell_p$ distances.
Coordinate medians minimize $\operatorname{SC}_1$, so
\[
    \operatorname{SC}_2(\mathbf{x},\operatorname{CM}(\mathbf{x}))
    \le\operatorname{SC}_1(\mathbf{x},\operatorname{CM}(\mathbf{x}))
    \le\operatorname{SC}_1(\mathbf{x},y^*)
    \le\sqrt{2}\,\operatorname{SC}_2(\mathbf{x},y^*).
\]
Finally, for $n=2$, the mechanism $f(x_1,x_2)=x_1$ is SP and has social cost $\|x_1-x_2\|$, which is optimal by the triangle inequality.
\end{proof}

\subsection{Exact Optimization in $\ell_\infty$ Spaces}\label{sec:single-endpoints}

In $\ell_\infty^d$, selecting the lexicographically smallest social-cost minimizer gives an anonymous SP mechanism with approximation ratio $1$.
This improves on CM, whose tight uniform approximation ratio in $\ell_\infty^2$ is $2$~\cite{hastings2026strategic}.

Define
\[
    f_{\mathrm{lex}}(\mathbf{x})
    =\operatorname{lexmin}\left(\argmin_{y\in\mathbb{R}^d}\sum_{i=1}^n\|x_i-y\|_\infty\right).
\]
The lexicographic rule first minimizes the first coordinate among all social-cost minimizers, then the second among the remaining points, and so on, in a fixed coordinate order.
The objective is continuous and tends to infinity as $\|y\|_\infty\to\infty$.
Its minimizer set is therefore nonempty and compact, so this successive minimization determines a unique output.

To prove SP, we compare the outputs before and after a possible misreport.
The following exchange lemma moves the outputs towards each other along selected coordinates.
For each unchanged report, the two distance changes sum to at most zero.
We choose the threshold so that the movement towards the second output reduces the manipulating agent's true cost.

\begin{lemma}\label{lem:single-linfty-exchange}
For $y,z\in\mathbb{R}^d$ and $\tau>0$, define
\[
    v_r=\begin{cases}
        \operatorname{sgn}(z_r-y_r),&|z_r-y_r|\ge\tau,\\
        0,&|z_r-y_r|<\tau.
    \end{cases}
\]
For each fixed $x\in\mathbb{R}^d$ and all sufficiently small $\varepsilon>0$,
\begin{equation}\label{eq:single-exchange}
    \|y+\varepsilon v-x\|_\infty+\|z-\varepsilon v-x\|_\infty
    \le\|y-x\|_\infty+\|z-x\|_\infty.
\end{equation}
If $\|a-y\|_\infty>\|a-z\|_\infty$ and $\tau=\|a-y\|_\infty-\|a-z\|_\infty$, then $v\ne0$ and, for all sufficiently small $\varepsilon>0$,
\begin{equation}\label{eq:single-exchange-gain}
    \|y+\varepsilon v-a\|_\infty=\|y-a\|_\infty-\varepsilon.
\end{equation}
\end{lemma}
\begin{proof}
The claim is immediate if $v=0$, so assume otherwise.
Let $A$ and $B$ be the right directional derivatives of the two distances in~\eqref{eq:single-exchange}.
Each belongs to $\{-1,0,1\}$, since every coordinate moves at speed zero or one and the norm is the maximum of the absolute coordinate differences.
We show that $A=1$ implies $B=-1$.
Write $r_y=\|y-x\|_\infty$ and $r_z=\|z-x\|_\infty$.
If $r_y>0$ and $A=1$, a coordinate attaining $r_y$ moves away from $x$ and is selected by the threshold rule.
In that coordinate the distance at $z$ is at least $r_y+\tau$, so $r_z\ge r_y+\tau$.
If $r_y=0$, the same inequality follows from $x=y$ and $v\ne0$.
For every coordinate $r$ attaining $r_z$,
\[
    |z_r-y_r|\ge |z_r-x_r|-|y_r-x_r|\ge r_z-r_y\ge\tau.
\]
All such coordinates are selected, and movement from $z_r$ towards $y_r$ reduces their distance to $x_r$ at rate one.
Thus $B=-1$.
Interchanging $y,z$ proves that $B=1$ implies $A=-1$, whence $A+B\le0$ in all cases.
Each distance is piecewise linear along the indicated movement, so this derivative inequality gives~\eqref{eq:single-exchange} for sufficiently small positive $\varepsilon$.

For the final assertion, put $R=\|a-y\|_\infty$, so $\|a-z\|_\infty=R-\tau$.
Every coordinate attaining $R$ satisfies $|a_r-y_r|=R$ and $|a_r-z_r|\le R-\tau$.
Hence $|z_r-y_r|\ge\tau$, and the direction from $y_r$ to $z_r$ points towards $a_r$.
Thus all coordinates attaining the maximum distance at $y$ move inwards at rate one.
For sufficiently small positive $\varepsilon$, the other coordinates do not overtake them.
This proves~\eqref{eq:single-exchange-gain} and also shows $v\ne0$.
\end{proof}

\begin{theorem}\label{thm:single-linfty}
For every $n\ge2$ and $d\ge2$, $f_{\mathrm{lex}}$ is an anonymous SP mechanism with approximation ratio $1$ in $\ell_\infty^d$.
\end{theorem}
\begin{proof}
Anonymity and optimality follow from the definition.
Suppose an agent with true location $a$ benefits by reporting $a'$.
Write $P=(a,\mathbf{x}_{-i})$, $P'=(a',\mathbf{x}_{-i})$, $y=f_{\mathrm{lex}}(P)$, and $z=f_{\mathrm{lex}}(P')$.
Then $\tau=\|a-y\|_\infty-\|a-z\|_\infty>0$.
Choose $v$ from Lemma~\ref{lem:single-linfty-exchange}.
Since there are finitely many reports, we can choose one $\varepsilon>0$ small enough for all its conclusions to hold.
We compare $y+\varepsilon v$ with $y$ on $P$, and $z-\varepsilon v$ with $z$ on $P'$.
The true agent's cost at $y+\varepsilon v$ decreases by $\varepsilon$.
The distance from the false report $a'$ to $z-\varepsilon v$ increases by at most $\varepsilon$, by the triangle inequality and $\|v\|_\infty=1$.
For every unchanged report, inequality~\eqref{eq:single-exchange} bounds the sum of its two cost changes by zero.
The manipulating agent's two cost changes sum to at most $-\varepsilon+\varepsilon=0$, and each other agent's pair of changes is also nonpositive.
Summing over agents yields
\[
    \operatorname{SC}(P,y+\varepsilon v)-\operatorname{SC}(P,y)
    +\operatorname{SC}(P',z-\varepsilon v)-\operatorname{SC}(P',z)\le0.
\]
Both cost differences are nonnegative because $y,z$ minimize their respective social costs.
Hence both are zero, and the moved points remain optimal for their respective profiles.
Because $v\ne0$, it has a first nonzero coordinate.
All earlier coordinates of both outputs remain unchanged.
If this coordinate of $v$ is negative, $y+\varepsilon v$ is lexicographically smaller than $y$.
If it is positive, $z-\varepsilon v$ is lexicographically smaller than $z$.
Either case contradicts the tie-breaking rule.
Hence no agent can benefit by misreporting.
\end{proof}

The mechanism can be computed by linear programming.
Introduce variables $y\in\mathbb{R}^d$ and $r_i\ge0$, minimize $\sum_i r_i$ subject to $-r_i\le y_j-x_{i,j}\le r_i$ for every $i,j$, and fix the optimal objective value.
Then minimize and fix $y_1,\ldots,y_d$ in order.
Thus $d+1$ linear programs suffice to compute $f_{\mathrm{lex}}$.

\section{Two-Facility Location}\label{sec:two-facilities}

We now consider two facilities in $\ell_p^d$, where $1<p<\infty$.
In the plane, no deterministic SP mechanism has a finite approximation ratio.
In higher dimensions, we obtain the same conclusion under anonymity, and in fact only require invariance under exchanging one pair of agents.
Both results hold for each fixed $n\ge3$.
The proof first reduces to three agents and studies the facility locations each agent can induce while another report is fixed.
A local restriction on these locations will force one agent to be served at her report throughout the fixed-report problem.
Exchange symmetry yields this restriction in every dimension.
In the plane, we also rule out the remaining possibility, in which the two agents control complementary lines.

\subsection{Common Structural Lemmas}\label{sec:two-tools}

Throughout this section, $\|\cdot\|$ denotes $\|\cdot\|_p$.
For a nonempty set $S$, write $\operatorname{dist}(x,S)=\inf_{y\in S}\|x-y\|$.
We use the duality map
\[
    J_p(v)=\bigl(|v_r|^{p-2}v_r\bigr)_{r=1}^d,
    \qquad \nabla\|v\|^p=pJ_p(v),
\]
where a zero coordinate of $J_p(v)$ is interpreted as zero.
Thus $J_p(v)$ describes how the $p$th power of distance changes under a small displacement.
For $p=2$, it is simply $J_2(v)=v$.
We have $\langle J_p(v),v\rangle=\|v\|^p$ and $J_p(tv)=t^{p-1}J_p(v)$ for $t>0$.
The norm is strictly convex and differentiable away from zero, and its $p$th power is differentiable everywhere.

It suffices to study three agents.
The following reduction uses partial group strategyproofness~\cite[Lemma 2.1]{lu10mechanism} and preserves exchange symmetry between two agents.

\begin{lemma}\label{lem:two-replication}
If there is a deterministic SP $\rho$-approximation for $n\ge3$ agents and two facilities, there is one for three agents with the same approximation guarantee.
If the original mechanism is invariant under exchanging a specified pair of agents, the reduction can preserve that exchange symmetry.
\end{lemma}
\begin{proof}
Given an $n$-agent mechanism $f$, define $\widehat f(a,b,c)=f(a,\ldots,a,b,c)$, using $n-2$ copies of $a$.
Strategyproofness for $b,c$ is inherited directly.
For the first report, change the copies of $a$ to $a'$ one at a time.
Each SP inequality compares distances from $a$ and says that this distance cannot decrease.
Chaining the inequalities proves SP for the first virtual agent.

For three locations, the optimal two-facility social cost is
\[
    \OPT_2(a,b,c)=\min\{\|a-b\|,\|a-c\|,\|b-c\|\}.
\]
Indeed, two of the three locations must share a nearest facility, giving the lower bound by the triangle inequality, and serving a closest pair with one facility attains it.
The same optimum holds for multiplicities $(n-2,1,1)$.
The unweighted lower bound still applies.
If a closest pair includes $a$, place one facility at $a$ and the other at the third location.
If $b,c$ are a closest pair, place facilities at $a,b$.
In each case, only one unreplicated agent pays the closest-pair distance.
The unweighted cost of $\widehat f$ is at most the original weighted cost, so its ratio is at most $\rho$.
To preserve exchange symmetry, use the specified pair for the two unreplicated reports.
\end{proof}

Assume for contradiction that a three-agent SP mechanism $f$ has finite approximation ratio $\rho$.
Fix agent $k$ at $a$ and denote the other two agents by $i,j$.
We keep these identities fixed and write $f(a,b,c)$ for reports $a,b,c$ by $k,i,j$, respectively.
Define the menu of agent $i$, when agent $j$ reports $q$, by
\[
    S_i(q)=\bigcup_{b\in\mathbb{R}^d}f(a,b,q).
\]
Define $S_j(q)$ similarly by fixing agent $i$ at $q$ and varying agent $j$.
A menu collects every facility that can appear as that agent changes her report, following the image-set framework of~\cite[Section 3.1]{lu10mechanism}.
We suppress the fixed identity $k$ and location $a$ in this notation.
The nearest-point and closedness properties below are standard consequences of SP~\cite[Lemma 3.1 and Corollary 3.3]{lu10mechanism}.

\begin{lemma}\label{lem:two-menus}
The menus are closed, contain $a,q$, and satisfy
\[
    \operatorname{dist}(b,f(a,b,q))=\operatorname{dist}(b,S_i(q))
    \le\rho\|a-q\|\qquad\text{for every }b.
\]
If $y$ is a nearest point of $S_i(q)$ to $b$, then reporting any point in the open segment from $b$ to $y$, or reporting $y$ itself, forces the output to contain $y$.
\end{lemma}
\begin{proof}
SP says that the truthful output attains the infimum distance over all reports, which gives the equality.
If $b$ belongs to the closure of the menu, this distance is zero, so the truthful output contains $b$.
Thus the menu is closed.
Profiles supported on $\{a,q\}$ have optimum zero, forcing a facility at each distinct reported location and hence $a,q\in S_i(q)$.
Placing facilities at $a,b$ costs at most $\|a-q\|$, giving the uniform bound.

For segment persistence, let $b_t=(1-t)b+ty$, where $0<t<1$.
For any $z\in S_i(q)$, the triangle inequality gives
\[
    \|b_t-z\|\ge\|b-z\|-\|b-b_t\|
    \ge(1-t)\|b-y\|=\|b_t-y\|.
\]
If $z\ne y$, equality would require equality in the triangle inequality and $\|b-z\|=\|b-y\|$.
Strict convexity then forces $z=y$, a contradiction.
Thus $y$ is the unique nearest menu point to $b_t$ and must be selected.
Reporting $y$ gives zero cost and has the same conclusion, as in the outcome-replacement property of~\cite[Corollary 3.2]{lu10mechanism}.
We refer to the open-segment conclusion as segment persistence.
\end{proof}

We use separated instances with $b,c$ close together and far from $a$, following the approach of Fotakis and Tzamos~\cite[Section 2]{fotakis2014power}.
If $D=\|b-c\|>0$ and $\min\{\|a-b\|,\|a-c\|\}>(\rho+1)D$, the approximation guarantee prevents $a$ from sharing a nearest facility with either $b$ or $c$.
Indeed, their combined distance to that facility would exceed $\rho D$, whereas $\OPT_2(a,b,c)\le D$.
There is therefore one facility serving $a$ and one facility serving both $b,c$.
Neither close agent can tie between the two facilities, since such a tie would also give a shared nearest facility with $a$.
Write $g(b,c)$ for their common facility on these separated profiles.
If either close agent reports the other's position, the resulting two-location profile has zero optimum.
The SP inequalities therefore give
\[
    \|b-g(b,c)\|\le\|b-c\|,
    \qquad \|c-g(b,c)\|\le\|b-c\|.
\]

For $q\ne a$, set $r(q)=\frac{\|a-q\|}{1000(\rho+2)}$.
This radius leaves room for the local report movements below while maintaining separation.
In particular, whenever $b,c\in B(q,4r(q))$, their distance is at most $8r(q)$ and their distances from $a$ are at least $\|a-q\|-4r(q)$.

\begin{lemma}\label{lem:two-local-geometry}
For sufficiently close $b,c$ away from $a$, their common facility $y=g(b,c)$ satisfies
\begin{equation}\label{eq:two-dual-geometry}
    \langle J_p(b-y),c-y\rangle\le0,
    \qquad \langle J_p(c-y),b-y\rangle\le0.
\end{equation}
Furthermore, each local menu is star-shaped about the fixed close report:
\[
    z\in S_i(q)\cap\overline B(q,r(q))
    \quad\Longrightarrow\quad [q,z]\subseteq S_i(q).
\]
The same statement holds for agent $j$.
\end{lemma}
\begin{proof}
Move $b$ towards $y$ to $z_t=y+t(b-y)$, with $t>0$ small.
Lemma~\ref{lem:two-menus} preserves the facility $y$.
All these reports remain in a sufficiently small neighborhood away from $a$, so $c$ is still served by $y$.
Agent $j$ can misreport $z_t$, in which case zero optimal cost forces facilities at $a,z_t$.
SP therefore gives $\|c-y\|\le\|c-z_t\|$.
The function $t\mapsto\|c-y-t(b-y)\|^p$ therefore has a nonnegative right derivative at zero.
Its derivative is $-p\langle J_p(c-y),b-y\rangle$, proving the second inequality in~\eqref{eq:two-dual-geometry}.
Interchanging the roles in this argument proves the first.

For star-shapedness, write $v=z-q$ and $h(t)=\operatorname{dist}(q+tv,S_i(q))^p$ for $0\le t\le1$.
Let $y_t$ be the common facility when the reports are $q+tv,q$.
The first inequality in~\eqref{eq:two-dual-geometry} gives
\[
    \langle J_p(q+tv-y_t),y_t-q\rangle\ge0.
\]
The function $h$ is Lipschitz on $[0,1]$.
At a differentiability point $t$, keep the nearest point $y_t$ fixed and compare $h(t+\eta)$ with $\|q+(t+\eta)v-y_t\|^p$.
Equality holds at $\eta=0$, and the upper bound for both signs of $\eta$ gives $h'(t)=p\langle J_p(q+tv-y_t),v\rangle$.
Consequently,
\[
    \frac{t}{p}h'(t)
    =h(t)+\langle J_p(q+tv-y_t),y_t-q\rangle\ge h(t).
\]
On every interval bounded away from zero, $h(t)/t^p$ is absolutely continuous and its derivative is nonnegative almost everywhere.
It is therefore nondecreasing for $t>0$.
Since $h(1)=0$, it vanishes throughout $(0,1]$, proving the claim.
\end{proof}

A local menu direction at $q$ is a unit vector parallel to $z-q$ for some $z\in S_i(q)\cap\overline B(q,r(q))$ with $z\ne q$.
Let $D_i(q)$ be the closure of these directions on the unit sphere, and let $C_i(q)$ be their closed convex conical hull.
If there are no such directions, set $D_i(q)=\varnothing$ and $C_i(q)=\{0\}$.
The cone $C_i(q)$ records the local directions in which agent $i$ can induce a facility.
Star-shapedness means that a realized direction remains available at arbitrarily small positive distances.
Thus shrinking the radius in this definition leaves the direction sets and cones unchanged.
We call a cone \emph{pointed} if it contains no nonzero line, equivalently if $C\cap(-C)=\{0\}$.

The following lemma transfers a nearest-point inequality from $S_i(c)$ to the local directions of $S_i(y)$.
This lets us compare menus with different fixed close reports.

\begin{lemma}\label{lem:two-normal-transfer}
Fix any $c$ and let $y$ be a nearest point of $S_i(c)$ to $x\notin S_i(c)$.
If $y\ne a$, then
\begin{equation}\label{eq:two-normal-transfer}
    \langle J_p(x-y),v\rangle\le0\qquad\text{for every }v\in C_i(y).
\end{equation}
\end{lemma}
\begin{proof}
Let agent $i$ report $z=y+\varepsilon(x-y)$ for sufficiently small $\varepsilon>0$.
Segment persistence forces a facility at $y$.
Agent $j$, if truly located at $y$, can obtain zero cost by reporting $c$.
Thus the profile with reports $z,y$ must also contain $y$.
These reports form a separated close pair, so $y$ is their common facility and is a nearest point of $S_i(y)$ to $z$.
For every realized local direction $v$, star-shapedness supplies $y+tv\in S_i(y)$ for all sufficiently small $t>0$.
Hence $\|z-y\|^p\le\|z-y-tv\|^p$.
Differentiating at zero gives $\langle J_p(x-y),v\rangle\le0$.
This inequality is linear in $v$, so it extends to the closed convex conical hull.
\end{proof}

We next compare the directions available to the two agents at the same center.
For Euclidean distance, the conclusion says that these directions form nonobtuse angles.

\begin{lemma}\label{lem:two-direction-compatibility}
For $u\in D_i(q)$ and $v\in D_j(q)$,
\begin{equation}\label{eq:two-direction-compatibility}
    \langle J_p(u),v\rangle\ge0,
    \qquad \langle J_p(v),u\rangle\ge0.
\end{equation}
\end{lemma}
\begin{proof}
First take realized unit directions $u,v$.
For sufficiently small $t,s>0$, star-shapedness gives $g(q+tu,q)=q+tu$ and $g(q,q+sv)=q+sv$.
Fix $t$ small enough that $2t<r(q+tu)$.
The point $q+tu$ is the unique nearest point of $S_j(q+tu)$ to $q$.
Indeed, another nearest point would lie within distance $2t$ of $q+tu$, and the segment between the two points would belong to the menu by star-shapedness.
Its midpoint would be strictly closer to $q$, by strict convexity.

It follows that $g(q+tu,q+sv)\to q+tu$ as $s\downarrow0$.
Indeed, nearest-point comparison with $q+tu$ bounds these facilities uniformly as $s\downarrow0$.
Every subsequential limit lies in the closed menu $S_j(q+tu)$ and minimizes the distance from $q$.
Uniqueness identifies every limit with $q+tu$.

Write $y_s=g(q+tu,q+sv)$.
For sufficiently small $t,s$, star-shapedness puts $[q+sv,y_s]$ in $S_i(q+sv)$.
When agent $i$ is at $q$, her truthful cost with this other report is $s$.
SP therefore gives
\[
    s^p\le\|sv+\lambda(y_s-q-sv)\|^p\qquad(0\le\lambda\le1).
\]
Differentiating at $\lambda=0$ yields $\langle J_p(v),y_s-q\rangle\ge s$.
Letting $s\downarrow0$ gives $\langle J_p(v),u\rangle\ge0$.
Fixing $s$ and letting $t$ tend to zero gives the other inequality.
Continuity of $J_p$ extends both inequalities to the closures defining $D_i(q),D_j(q)$.
\end{proof}

For each fixed argument of $J_p$, linearity in the other argument extends~\eqref{eq:two-direction-compatibility} to conical combinations.

The next lemma is the main propagation step.
It turns a local restriction on the directions available to one agent into exact service for the other throughout the problem with agent $k$ fixed at $a$.
Here exact service means that a facility is placed at the agent's report.

\begin{lemma}\label{lem:two-pointed}
If $C_i(q)$ is pointed for some $q\ne a$, then agent $j$ is always served at her reported location when agent $k$ is fixed at $a$:
\[
    c\in f(a,b,c)\qquad\text{for every }b,c\in\mathbb{R}^d.
\]
\end{lemma}
\begin{proof}
We first show that selecting the inner report persists as both reports move along a ray.
We then use this property to show that the other agent's menus contain an open cone and, finally, the whole space.
Propagation of endpoint choices on separated line instances is used in~\cite[Section 2.1 and Appendix A]{fotakis2014power}.
Here the common facility can leave the line of the reports, so we first establish a norm inequality controlling that possibility.
We need the following elementary strict inequality: there is a constant $A_p\ge1$ such that
\begin{equation}\label{eq:two-strict-step}
    \langle J_p(x),h\rangle\ge0,\quad h\ne0,\quad 0<\theta<\frac{1}{A_p}
    \quad\Longrightarrow\quad \|x-\theta h\|<\|x+h\|.
\end{equation}
To prove this, set $R(x,h)=\|x+h\|^p-\|x\|^p-p\langle J_p(x),h\rangle$.
Strict convexity gives $R(x,h)>0$ for $h\ne0$.
In one coordinate, the ratio $R(s,-1)/R(s,1)$ is continuous in $s$, has positive denominator, and tends to $1$ as $s\to\pm\infty$.
It therefore has a finite supremum $A_p$.
Scaling each coordinate and summing gives $R(x,-h)\le A_pR(x,h)$ in every dimension.
Convexity in the step length gives $R(x,-\theta h)\le\theta A_pR(x,h)$ for $0\le\theta\le1$.
Comparing the linear terms and using $\theta A_p<1$ proves~\eqref{eq:two-strict-step}.

Consider a unit direction $u$ whose ray $q+tu$, $t\ge0$, stays a positive distance from $a$.
Choose a step $\delta>0$ smaller than a fixed fraction of $r$ everywhere along this ray.
Suppose $g(q+\delta u,q)=q$, with agent $i$ at the outer report.
At a current inner report $b$, put $c=b+\delta u$ and assume $g(c,b)=b$.
Moving the inner report to $b'=b+\eta u$, where $0<\eta<\frac{\delta}{A_p+1}$, selects $b'$: the menu for agent $j$ contains $b,c$, and its local star-shapedness contains their entire segment.
Next move the outer report from $c$ to $c'=b+(\delta+\eta)u$, and write $y=g(c',b')$.
If $y\ne b'$, apply~\eqref{eq:two-strict-step} with $x=c'-y$, $h=y-b'$, and $\theta=\frac{\eta}{\delta-\eta}$.
The premise on the dual inner product follows from~\eqref{eq:two-dual-geometry}.
Since $x+h=\delta u$ and $c-y=(1-\frac{\eta}{\delta})(x-\theta h)$, we obtain
\[
    \|c-y\|<\left(1-\frac{\eta}{\delta}\right)\delta=\delta-\eta.
\]
This is a beneficial deviation for agent $i$, whose previous truthful cost was $\delta-\eta$.
Thus $y=b'$, and the endpoint selection persists after translating both reports by $\eta u$.
The step size can be chosen uniformly along the ray because its distance from $a$ has a positive lower bound.
Repeating steps of size less than $\frac{\delta}{A_p+1}$, with a shorter final step if needed, reaches any specified distance along the ray.
Segment persistence then shortens the separation, giving
\begin{equation}\label{eq:two-ray-selection}
    g(q+(t+s)u,q+tu)=q+tu\qquad(t\ge0,\ 0<s\le\delta).
\end{equation}

Because $C_i(q)$ is pointed, its polar cone $\{w:\langle w,v\rangle\le0\text{ for all }v\in C_i(q)\}$ has nonempty interior.
The map $J_p$ is a continuous bijection with continuous inverse.
We can therefore choose a small open set of unit directions $U$ whose images under $J_p$ lie in that interior.
We may shrink $U$ so that its closure avoids the direction from $q$ to $a$.
For $u\in U$ and small $t>0$, the unique nearest point of $S_i(q)$ to $q+tu$ is $q$.
Indeed, for every nonzero local menu displacement $v$, convexity gives
\[
    \|tu-v\|^p\ge\|tu\|^p-p\langle J_p(tu),v\rangle>\|tu\|^p,
\]
and menu points outside the local ball are farther away when $t$ is small.
This also holds when $C_i(q)=\{0\}$, since there are no other local menu points.
The rays in this smaller set stay a uniform positive distance from $a$, so a common sufficiently small step in~\eqref{eq:two-ray-selection} works for all of them.
Let $\Gamma=\{q+tu:t>0,\ u\in U\}$.
For $y=q+tu\in\Gamma$, choose $0<s<\min\{t,\delta\}$.
Equation~\eqref{eq:two-ray-selection} gives $g(y,y-su)=y-su$.
Thus, with agent $i$ fixed at $y$, agent $j$ can induce points in the direction towards $q$.
Hence
\[
    q-y\in C_j(y)\qquad(y\in\Gamma).
\]

Fix any report $c$ of agent $i$, and write $S=S_j(c)$.
Along an interior ray of $\Gamma$, set $h(t)=\operatorname{dist}(q+tu,S)^p$.
Lemma~\ref{lem:two-menus} bounds $h$ by $\rho^p\|a-c\|^p$.
Whenever $h(t)>0$ and a nearest menu point $y_t$ lies in $\Gamma$, Lemma~\ref{lem:two-normal-transfer} gives $\langle J_p(q+tu-y_t),y_t-q\rangle\ge0$.
At differentiability points, it follows that
\[
    \frac{t}{p}h'(t)
    =h(t)+\langle J_p(q+tu-y_t),y_t-q\rangle\ge h(t).
\]
As in the star-shapedness argument, absolute continuity makes $h(t)/t^p$ nondecreasing on any interval where these conditions hold.
We use this monotonicity to rule out $h(t)>0$.

Suppose $h(t_0)>0$.
If the ray first returns to $S$ at a finite time $T>t_0$, then $h(t)>0$ for $t_0\le t<T$ and $h(T)=0$.
For any nearest point $y_t$, the bound $\|y_t-(q+Tu)\|\le h(t)^{1/p}+T-t$ shows that all nearest points tend to $q+Tu$ as $t\uparrow T$.
They therefore lie in $\Gamma$ just before $T$.
The preceding monotonicity would force a positive value of $h(t)/t^p$ to remain positive at $T$, a contradiction.
If the ray never returns, all nearest points lie in $\Gamma$ for sufficiently large $t$: their distances to the ray are uniformly bounded, whereas the ray's distance from the boundary of $\Gamma$ grows linearly.
On this tail, $h(t)/t^p$ is positive and nondecreasing, but is bounded above by $\rho^p\|a-c\|^p/t^p\to0$, again a contradiction.
Consequently $\Gamma\subseteq S_j(c)$ for every $c$.

Choose $q_0\in\Gamma$.
Agent $j$ is exactly served at $q_0$ against every report of agent $i$.
In particular, for every unit direction $u$ and sufficiently small $s>0$, the separated profile $(q_0+su,q_0)$ has common facility $q_0$.
This supplies the initial endpoint choice for ray propagation in every direction.
Repeat the ray propagation and menu-filling argument for any direction other than the one from $q_0$ towards $a$, taking a small surrounding set of directions each time.
This fills a dense subset of $\mathbb{R}^d$ in every $S_j(c)$.
Closedness gives $S_j(c)=\mathbb{R}^d$, so SP forces exact service at every report of agent $j$.
\end{proof}

\subsection{Lower Bounds in the Plane}\label{sec:two-planar}

\begin{theorem}\label{thm:two-planar}
For every $1<p<\infty$ and every $n\ge3$, no deterministic SP two-facility mechanism in $\ell_p^2$ has a finite social-cost approximation ratio.
\end{theorem}

In the plane, if neither local direction cone is pointed, the two agents must control complementary lines.
The following lemma makes this statement precise and shows that the lines cannot vary with the location of the close pair.

\begin{lemma}\label{lem:two-planar-lines}
Fix agent $k$ at $a$.
Either one of the other agents is exactly served at every profile with this fixed report, or there are fixed distinct lines $E_i,E_j$ through zero such that
\[
    \langle J_p(e),v\rangle=\langle J_p(v),e\rangle=0
    \qquad(e\in E_i,\ v\in E_j),
\]
and, for every $q\ne a$, the local menu of agent $i$ is the full line segment
\[
    S_i(q)\cap\overline B(q,r(q))=(q+E_i)\cap\overline B(q,r(q)).
\]
The analogous identity holds for agent $j$ with line $E_j$.
\end{lemma}
\begin{proof}
If some direction cone is pointed, Lemma~\ref{lem:two-pointed} gives the first alternative.
Otherwise both $C_i(q),C_j(q)$ contain a line at every $q\ne a$.
Choose a nonzero vector $e$ on a line contained in $C_i(q)$.
For each $v\in D_j(q)$, Lemma~\ref{lem:two-direction-compatibility}, extended linearly in its second argument, gives $\langle J_p(v),e\rangle=0$.
In two dimensions, $e^\perp$ is a line.
Since $J_p$ is an odd homogeneous bijection, the inverse image $J_p^{-1}(e^\perp)$ is also a line.
Since $C_j(q)$ contains a line, it must equal this line.
Its generating directions include both signs, so applying the same inequalities again forces $C_i(q)$ to be a line as well.
Denote them by $E_i(q),E_j(q)$.
Both dual orthogonality identities follow, and the lines are distinct because $\langle J_p(e),e\rangle>0$ for $e\ne0$.

Star-shapedness and closedness show that each local menu is a closed interval on its line, with positive length on both sides of $q$.
We show that it fills the prescribed local ball.
Suppose, for a unit $e\in E_i(q)$, its positive side ends at $z=q+\ell e$ with $0<\ell<r(q)$.
At reports $z,q$, the common facility is $z$.
Agent $j$ at $q$ therefore has nearest menu point $z$ in $S_j(z)$.
Testing small displacements in both directions of $E_j(z)$ gives $\langle J_p(q-z),v\rangle=0$ for $v\in E_j(z)$.
Thus $E_j(z)=\ker J_p(e)=E_j(q)$ and $E_i(z)=E_i(q)$.
For sufficiently small $\varepsilon>0$, the point $z$ is the unique nearest point of $S_i(q)$ to $x=z+\varepsilon e$.
The local interval ends at $z$, while points outside the local ball are excluded by choosing $\varepsilon<\frac{r(q)-\ell}{2}$.
Lemma~\ref{lem:two-normal-transfer}, applied with nearest point $z$, now requires $\langle J_p(\varepsilon e),e\rangle\le0$, a contradiction.
The negative side and the other agent are treated in the same way.

Projection onto these mutually orthogonal lines is linear and complementary.
Indeed, for nonzero $e\in E_i(q),v\in E_j(q)$, a vector written as $\alpha e+\beta v$ projects onto $E_i(q)$ at $\alpha e$: the derivative in direction $e$ vanishes at residual $\beta v$, and strict convexity gives uniqueness.
The other projection is $\beta v$.
For sufficiently close $b,c$ near any fixed $q\ne a$, the common facility therefore has both representations
\[
    g(b,c)=c+P_{E_i(c)}(b-c)=b+P_{E_j(b)}(c-b)
    =c+P_{E_i(b)}(b-c).
\]
If the common projection of $b-c$ is nonzero, it spans both $E_i(b)$ and $E_i(c)$.
If it is zero and $b\ne c$, the vector $b-c$ spans both complementary kernel lines, which also determines $E_i(b)=E_i(c)$.
Thus the lines are locally constant.
Since $\mathbb{R}^2\setminus\{a\}$ is connected, they are constant throughout this domain.
\end{proof}

We now compare the line structures obtained by fixing different agents.
The next lemma transfers one agent's line to a new fixed-report problem.
Following two different sequences of transfers will then give incompatible lines for the same problem.
To compare these choices, write $E_i^k(a)$ for the line of agent $i$ when agent $k$ is fixed at $a$.
This notation is only used for restrictions of the second type in Lemma~\ref{lem:two-planar-lines}.

\begin{lemma}\label{lem:two-line-transfer}
Suppose $E_i^k(a)=E$, and let $c-a\in E\setminus\{0\}$.
Then the restriction fixing agent $j$ at $c$ is also of the second type, and $E_i^j(c)=E$.
\end{lemma}
\begin{proof}
Consider agent $i$'s menu with the other reports fixed at $a,c$.
In the restriction fixing $k$ at $a$, Lemma~\ref{lem:two-normal-transfer} says that every nearest-point residual with nearest point $y\ne a$ satisfies $\langle J_p(x-y),e\rangle=0$ for $e\in E$.
This forces the menu to contain the entire segment $(a,c]$.
To prove this, consider the $p$th power of the distance to the menu along any compact subsegment of $(a,c]$.
Near a zero-distance point $z\ne a$, every nearest point $y$ satisfies $\|y-z\|\le2\|x-z\|$, because $z$ itself belongs to the menu.
Thus these nearest points stay away from $a$.
At positive distance, the nearest-point derivative formula and $\langle J_p(x-y),e\rangle=0$ give zero derivative along $E$.
At zero distance, the derivative of the $p$th power is also zero because $p>1$.
Absolute continuity therefore makes this distance function constant on a neighborhood of $z$.
The set of zero-distance points on the segment is consequently both closed and relatively open.
Since it contains $c$, it is the whole segment.

Now view the same menu in the restriction fixing $j$ at $c$.
If agent $i$ were always exactly served, the menu would be the whole plane, contradicting its local line structure near $c$ in the original restriction.
If agent $k$ were always exactly served, take $z\in(a,c]$ arbitrarily close to $a$.
At reports $a,z,c$, exact service of $k$ and membership of $z$ in agent $i$'s menu force the two facilities to be $a,z$.
The cost of agent $j$ then tends to $\|c-a\|>0$, while the optimum is at most $\|z-a\|\to0$.
This contradicts finite approximation.
Thus the new restriction has the line structure of Lemma~\ref{lem:two-planar-lines}.
Its menu near $a$ contains points of $(a,c]$, so its line for agent $i$ is $E$.
\end{proof}

\begin{proof}[Proof of Theorem~\ref{thm:two-planar}]
By Lemma~\ref{lem:two-replication}, it suffices to rule out a finite approximation for three agents.
Suppose first that one fixed-report restriction has the line structure, and write $E_1^3(a)=E$, $E_2^3(a)=F$.
Choose nonzero $e\in E$ and $v\in F$.
Apply Lemma~\ref{lem:two-line-transfer} from $a$ to $a+e$, fixing agent $2$ there.
Agent $1$ retains line $E$, so agent $3$ has the complementary line $F$.
Transfer agent $3$'s line first to a restriction fixing agent $1$ at $a+e+v$, and then to one fixing agent $2$ at $a+e+2v$.
This gives $E_3^2(a+e+2v)=F$.
On the other hand, start with $E_2^3(a)=F$ and transfer to the restriction fixing agent $1$ at $a+2v$.
There agent $3$ has line $E$.
Transferring along $e$ gives $E_3^2(a+e+2v)=E$, a contradiction.
The complementary line is unique in the plane, so all of these identifications use the same pair $E,F$.

Every fixed-report restriction must therefore have an agent who is always exactly served.
Fix agent $3$ at $a$, and relabel the others so that agent $1$ is always served.
Next fix agent $1$ at $b\ne a$.
If agent $3$ is always served in this restriction, then $f(b,x_2,a)=\{a,b\}$ for every $x_2$.
Sending $x_2$ to infinity makes the mechanism cost unbounded, whereas the optimum is at most $\|a-b\|$.
If agent $2$ is always served instead, then $f(b,x_2,a)=\{b,x_2\}$ for $x_2\ne b$.
Letting $x_2\to b$ makes the optimum tend to zero, while agent $3$'s cost tends to $\|a-b\|>0$.
Both cases contradict finite approximation.
\end{proof}

\subsection{Anonymous Lower Bounds in Higher Dimensions}\label{sec:two-anonymous}

In higher dimensions, exchange symmetry makes the two agents' direction sets coincide.
Lemma~\ref{lem:two-direction-compatibility} then forces their common cone to be pointed, allowing us to apply the propagation lemma directly.

\begin{theorem}\label{thm:two-anonymous}
For every finite $d\ge2$, every $1<p<\infty$, and every $n\ge3$, no anonymous deterministic SP two-facility mechanism in $\ell_p^d$ has a finite social-cost approximation ratio.
The same conclusion holds if anonymity is replaced by invariance under exchanging one fixed pair of agents.
\end{theorem}
\begin{proof}
Apply Lemma~\ref{lem:two-replication}, keeping the interchangeable pair as agents $i,j$ and grouping the remaining agents into agent $k$.
Fix $k$ at $a$ and take any $q\ne a$.
Exchange symmetry gives $D_i(q)=D_j(q)$.
If this common set is empty, its cone is $\{0\}$ and is pointed.
Otherwise denote it by $D$.
It is a nonempty compact subset of the unit sphere, and Lemma~\ref{lem:two-direction-compatibility} gives $\langle J_p(u),v\rangle\ge0$ for all $u,v\in D$.

We claim that $0\notin\operatorname{conv}D$.
If $\sum_{r=1}^m\lambda_ru_r=0$ with $u_r\in D$, $\lambda_r>0$, and $\sum_r\lambda_r=1$, then
\[
    0=\left\langle J_p(u_1),\sum_{r=1}^m\lambda_ru_r\right\rangle
    =\sum_{r=1}^m\lambda_r\langle J_p(u_1),u_r\rangle
    \ge\lambda_1>0,
\]
a contradiction.
Compactness of $\operatorname{conv}D$ therefore gives a linear functional $\ell$ and $\delta>0$ such that $\ell(u)\ge\delta$ for every $u\in D$.
For any nonnegative combination $w=\sum_r\beta_ru_r$,
\[
    \ell(w)\ge\delta\sum_r\beta_r\ge\delta\|w\|.
\]
The inequality persists under closure, so the common direction cone cannot contain both $w\ne0$ and $-w$.
It is pointed.

Lemma~\ref{lem:two-pointed} implies that agent $j$ is always exactly served when $k$ is fixed at $a$.
Exchange symmetry gives the same conclusion for agent $i$.
Thus $f(a,b,c)=\{b,c\}$ whenever $b\ne c$.
Take $b,c$ tending to the same point different from $a$.
Their distance, and hence the optimal cost, tends to zero, whereas the cost of agent $k$ stays bounded away from zero.
This contradicts a finite approximation ratio.
\end{proof}

These proofs use strict convexity for segment persistence and local uniqueness, and the differentiable duality map for propagation.
The cases $p=1,\infty$ remain open.

\section{Extension to Multiple Facilities}\label{sec:multiple-facilities}

The two-facility lower bounds extend to every $k\ge3$ with $n>k$.
We place $k-2$ additional agents far from the remaining agents and from one another.
As in the separated-instance arguments of~\cite[Section 2]{fotakis2014power}, a finite approximation guarantee forces each distant agent to be served by a distinct facility.
At most two facilities remain available to the original agents.
We turn these facilities into a two-facility mechanism by sending the additional reports to infinity and taking a common limit over all profiles.
This common limit preserves every SP inequality.

\begin{lemma}[Distant-agent reduction]\label{lem:multiple-reduction}
Let $X$ be a finite-dimensional normed space and $2\le k<n$.
If there is a deterministic SP $\rho$-approximation for $n$ agents and $k$ facilities in $X$, then there is a deterministic SP $\rho$-approximation for $n-k+2$ agents and two facilities in the same space.
The reduction preserves anonymity.
It also preserves invariance under exchanging any fixed pair of agents retained in the smaller instance.
\end{lemma}
\begin{proof}
The cases $k=2$ and $X=\{0\}$ are immediate, so suppose $k\ge3$ and $X\ne\{0\}$, and let $f$ be the original mechanism.
Retain $m=n-k+2$ agents, with arbitrary profile $\mathbf{x}=(x_1,\ldots,x_m)$.
Fix a unit vector $e$ and place the additional agents at $Re,2Re,\ldots,(k-2)Re$, where $R$ ranges over positive integers.
Write $\mathbf{x}^{R}=(\mathbf{x},Re,\ldots,(k-2)Re)$ for the augmented profile.
Placing facilities at all $k-2$ additional reports and at an optimal pair for $\mathbf{x}$ shows that $\OPT_k(\mathbf{x}^R)\le\OPT_2(\mathbf{x})$.
The approximation guarantee of $f$ therefore gives
\begin{equation}\label{eq:multiple-cost}
    \operatorname{SC}\bigl(\mathbf{x}^{R},f(\mathbf{x}^{R})\bigr)
    \le\rho\,\OPT_2(\mathbf{x}).
\end{equation}
For a fixed $\mathbf{x}$, this bound is independent of $R$.
Thus every additional agent has a facility within distance $\rho\,\OPT_2(\mathbf{x})$.
For sufficiently large $R$, these $k-2$ neighborhoods are disjoint and lie outside $\overline B(0,\frac R3)$.
At most two facilities can therefore lie inside this ball.
Meanwhile, each retained agent has a nearest facility within a bounded distance of her fixed report, so all such nearest facilities lie inside the ball for sufficiently large $R$.

Define $g_R(\mathbf{x})$ by retaining the facilities of $f(\mathbf{x}^{R})$ inside $\overline B(0,\frac R3)$.
If this leaves no facility or more than two, set $g_R(\mathbf{x})=\{0\}$.
For each fixed profile, the preceding bounds ensure that $g_R$ retains between one and two facilities for all sufficiently large $R$.
The required threshold may depend on the profile.
Equation~\eqref{eq:multiple-cost} then implies, for all sufficiently large $R$,
\begin{equation}\label{eq:multiple-core-cost}
    \operatorname{SC}(\mathbf{x},g_R(\mathbf{x}))\le\rho\,\OPT_2(\mathbf{x}).
\end{equation}

Each fixed incentive comparison also holds for all sufficiently large $R$.
To see this, fix a profile $\mathbf{x}$ and a deviation $x_i'$.
For both profiles $\mathbf{x}$ and $(x_i',\mathbf{x}_{-i})$, all retained reports have bounded service distances as $R\to\infty$.
At the deviating profile, a facility near $x_i'$ is also at a bounded distance from the true location $x_i$.
Every deleted facility, by contrast, is at distance at least $\frac R3-\|x_i\|$ from $x_i$.
Thus deleting the distant facilities preserves both costs in this fixed incentive comparison for all sufficiently large $R$.
SP of $f$ gives
\begin{equation}\label{eq:multiple-eventual-sp}
    \operatorname{dist}(x_i,g_R(\mathbf{x}))
    \le\operatorname{dist}(x_i,g_R(x_i',\mathbf{x}_{-i})).
\end{equation}

To preserve these inequalities, we need convergence along the same indices at every truthful and deviating profile.
We therefore regard each rule $g_R$ as a point in a product space, with one coordinate for each profile.
Use the one-point compactification $\overline X=X\cup\{\infty\}$, in which a sequence tends to $\infty$ precisely when it eventually leaves every bounded set.
Represent each $g_R(\mathbf{x})$ by an ordered pair in $\overline X^2$, sorting finite points by norm and breaking ties lexicographically in fixed coordinates.
If there is only one facility, use $\infty$ as the second entry.
Finite dimensionality makes closed bounded subsets of $X$ compact, and adjoining $\infty$ gives the compact space $\overline X$.
Tychonoff's theorem makes the product $\prod_{\mathbf{x}\in X^m}\overline X^2$ compact as well.
Consequently, the sequence of rules has a convergent subnet, giving convergence at every profile along the same choice of indices.
A subnet is the generalization of a subsequence used for compact spaces whose products may have uncountably many coordinates.
Its indices still satisfy $R\to\infty$, so every inequality that holds for all sufficiently large $R$ also holds eventually along this subnet.

At each fixed profile, at least one limit point is finite.
Indeed,~\eqref{eq:multiple-core-cost} gives a facility with norm at most $\|x_1\|+\rho\,\OPT_2(\mathbf{x})$.
The first entry has the smallest norm, so it stays in a fixed compact subset of $X$ and has a finite limit.
Let $g(\mathbf{x})$ be the set of finite entries in its limiting pair, deleting $\infty$ and any duplicate.
This is a nonempty set of at most two facilities.
For every fixed location $z\in X$, convergence of the pairs gives
\[
    \operatorname{dist}(z,g_R(\mathbf{x}))
    \longrightarrow\operatorname{dist}(z,g(\mathbf{x}))
\]
along the selected subnet: finite entries converge in norm, and entries tending to $\infty$ eventually become farther from $z$ than any finite limiting entry.
Applying this convergence at each $x_i$ and summing proves the approximation guarantee from~\eqref{eq:multiple-core-cost}.
Applying it at the true location $x_i$ for both profiles in~\eqref{eq:multiple-eventual-sp} proves SP.

Finally, truncating the facility set, assigning $\{0\}$ when needed, and ordering the outputs all preserve symmetries among the retained agents.
If two profiles differ by a permutation under which $f$ is invariant, their ordered pairs are equal for every $R$.
Their limiting pairs are therefore equal as well, proving both symmetry assertions.
\end{proof}

\begin{theorem}\label{thm:multiple-lower-bounds}
Let $1<p<\infty$ and $2\le k<n$.
In $\ell_p^2$, no deterministic SP $k$-facility mechanism has a finite social-cost approximation ratio.
In every finite dimension $d\ge2$, the same conclusion holds for anonymous deterministic SP mechanisms in $\ell_p^d$.
For the latter conclusion, invariance under exchanging one fixed pair of agents suffices.
\end{theorem}
\begin{proof}
A finite-approximation mechanism would, by Lemma~\ref{lem:multiple-reduction}, induce a two-facility mechanism for $n-k+2\ge3$ agents with the same finite guarantee.
In the plane this contradicts Theorem~\ref{thm:two-planar}.
For the higher-dimensional claim, the reduction preserves anonymity, contradicting Theorem~\ref{thm:two-anonymous}.
If only one pair is interchangeable, retain that pair among the $n-k+2$ agents and apply the stronger assertion of Theorem~\ref{thm:two-anonymous}.
\end{proof}

\section{Conclusion}\label{sec:conclusion}

We established optimality and impossibility results for deterministic strategyproof facility location under social cost.
For a single facility in the Euclidean plane, the coordinate-wise median achieves the best approximation ratio for every $n\ge3$, even without anonymity.
Its optimality depends on the geometry of the space: in $\ell_\infty^d$, selecting the lexicographically smallest social-cost minimizer is strategyproof and achieves ratio $1$.
For multiple facilities, we proved that no deterministic SP mechanism in $\ell_p^2$, where $1<p<\infty$, has a finite approximation ratio when $2\le k<n$.
In higher dimensions, invariance under exchanging one fixed pair of agents suffices for the same conclusion.
A distant-agent reduction transfers the two-facility lower bounds to larger numbers of facilities while preserving the approximation guarantee and strategyproofness.

Several questions remain for multidimensional spaces.
Can the symmetry assumption be removed from the higher-dimensional lower bound?
For the endpoint norms $p=1$ and $p=\infty$, do multiple facilities admit finite approximation guarantees?
On the real line, another question concerns $3\le k<n$.
Fotakis and Tzamos~\cite{fotakis2014power} ruled out finite approximation guarantees under anonymity.
Does a finite social-cost approximation become possible without anonymity?

\paragraph{AI Use Statement.}
After several unsuccessful attempts to prove the theorems for multiple facilities, the author used GPT Astra-6 to explore possible proofs.
Codex was also used to assist with drafting the paper and checking the proofs.
The author carefully read and checked the content and takes responsibility for its correctness.

\bibliographystyle{plain}
\bibliography{mybibfile}

\end{document}